\documentclass{llncs}
\usepackage[utf8]{inputenc}
\usepackage{amssymb}
\usepackage{amsmath}
\usepackage{graphicx}
\usepackage{comment}
\usepackage{soul}
\usepackage{listings}
\usepackage[russian,english]{babel}
\usepackage{babelbib}

\newtheorem{notation}{Notation}
\usepackage{tikz}
\usepackage{caption}
\usetikzlibrary{patterns}
\usepackage[toc,page]{appendix}
\title{Hardness and algorithmic results for the approximate cover problem}
\author{Alexandru Popa\inst{1,2} \and Andrei Tanasescu\inst{3}}

\institute{University of Bucharest \and
National Institute of Research and Development in Informatics \and Politehnica University of Bucharest \\
E-mail: \email{alexandru.popa@fmi.unibuc.ro, andrei.tanasescu@mail.ru}}

\begin{document}
\maketitle

\begin{abstract}
In CPM 2017, Amir et al. introduce a problem, named \emph{approximate string cover} (\textbf{ACP}), motivated by many aplications including coding and automata theory, formal language theory, combinatorics and molecular biology.  A \emph{cover} of a string $T$ is a string $C$ for which every letter of $T$ lies within some occurrence of $C$. The input of the \textbf{ACP} problem consists of a string $T$ and an integer $m$ (less than the length of $T$), and the goal is to find a string $C$ of length $m$ that covers a string $T'$ which is as close to $T$ as possible (under some predefined distance). Amir et al. study the problem for the Hamming distance.

In this paper we continue the work of Amir et al. and show the following results:

\begin{enumerate}
\item We show an approximation algorithm for the \textbf{ACP} with an approximation ratio of $\sqrt{OPT}$, where OPT is the size of the optimal solution.
\item We provide an FPT algorithm with respect to the alphabet size.
\item The \textbf{ACP} problem naturally extends to pseudometrics. Moreover, we show that for some family of pseudometrics, that we term \emph{homogenous additive pseudometrics}, the complexity of \textbf{ACP} remains unchanged.
\item We partially give an answer to an open problem of Amir et al. and show that the Hamming distance over an unbounded alphabet is equivalent to an extended metric over a fixed sized alphabet.
\end{enumerate}
\end{abstract}

\section{Introduction}

\paragraph*{\sc Motivation.}

Redundancy is a common trait of all natural data and was intensely studied over the years for its descriptive capabilities~\cite{ming1990kolmogorov,muchnik2003almost}. Errors can occur at any point in the data manipulation process, but by the use of redundancy they may be detected and, perhaps, corrected before propagation. 

Consider the transmission of a message over a radio frequency. Since we transmit over radio, we must use a digital to analog converter, that modulates our signal in amplitude and/or phase. In our example, we consider $amplitude\ shift\ keying$ (see, e.g.,~\cite{middlestead2017digital}). At the other end, the signal must be converted back, but we must check for transmission errors. If the channel is not too noisy, we may round the received amplitude to the spectrum we are using. However, we must be able to at least tell when it is too noisy. Since the signal we sent is smooth and periodic, we may  smooth our data and identify interference as unnatural spikes in the received input. This, however, only accounts for major interferences, and we cannot possibly do more at the physical level, since we may not assume smoothness of the sent data itself and must rely instead on the redundancy at some higher data, that is no longer agnostic to the message's form. 


Periodicity is a very important phenomenon when analyzing physical data such as an analogue signal. In general, natural data is very redundant or repetitive and exhibits some key patterns or regularities~\cite{HAVLIN1995171,timmermans2017cyclical,tychonoff1935theoremes}. Periodicity itself has been thoroughly studied in various fields such as Signal Processing~\cite{sethares1999periodicity}, Bioinformatics~\cite{brodzik2007quaternionic}, Dynamical Systems~\cite{katok1997introduction} and Control Theory~\cite{bacciotti2006liapunov}, each bringing its own insights. 

However some phenomena are  not periodical by nature, even if they are very redundant. Consider for instance the string $abaabaababa$: even though it is not periodic it clearly exhibits a single pattern, $aba$, and thus, we shall call it $quasi-periodic$ (see~\cite{ApostolicoB97}). Depending on the specific perturbations this may or may not be adequate. For example, $abaabaababa$ could be a repeated $aba$ that suffers from two $a$s so close together that they fuse (or some other desynchronization), as sounds sometimes do in natural language. In fact even $abaabaabcba$ and $abaabaabaaca$ exhibit the pattern $aba$ and the nonconforming $c$ could result from some echo or corruption. Depending on the task at hand we may want to retrieve either the information ($aba$) or the peculiarities in its transmission (the non-periodicity). 

\vspace*{-1cm}
\begin{figure}
\noindent\begin{minipage}{\textwidth}
\begin{minipage}[c][3cm][c]{\dimexpr0.5\textwidth-5pt\relax}
\resizebox{\textwidth}{!}{
\begin{tikzpicture}
	\draw[thick, black] (-1,0) -- (0,0);
    \draw[thick, black] (12.3,0) -- (13.3,0);

    \draw[thick, black] (0,0) sin (0.25,2) cos (0.5,0) sin (0.75,-2) cos (1,0) -- (1.1,0);
    \draw[thick, black] (1.1,0) sin (1.35,1) cos (1.6,0) sin (1.85,-1) cos (2.1,0) -- (2.2,0);
    \draw[thick, black] (2.2,0) sin (2.45,2) cos (2.7,0) sin (2.95,-2) cos (3.2,0) -- (3.4,0);
    
    \draw[thick, black] (3.4,0) sin (3.65,2) cos (3.9,0) sin (4.15,-2) cos (4.4,0) -- (4.5,0);
    \draw[thick, black] (4.5,0) sin (4.75,1) cos (5.0,0) sin (5.25,-1) cos (5.5,0) -- (5.6,0);
    \draw[thick, black] (5.6,0) sin (5.85,2) cos (6.1,0) sin (6.35,-2) cos (6.6,0) -- (6.8,0);
    
    \draw[thick, black] (6.8,0) sin (7.05,2) cos (7.3,0) sin (7.55,-2) cos (7.8,0) -- (7.9,0);
    \draw[thick, black] (7.9,0) sin (8.15,1) cos (8.4,0) sin (8.65,-1) cos (8.9,0) -- (9.0,0);
    \draw[thick, black] (9.0,0) sin (9.25,2) cos (9.5,0) sin (9.75,-2) cos (10.0,0) -- (10.1,0);
    
    \draw[thick, dashed, black] (9.0,0) -- (9.1,0) sin (9.35,2) cos (9.6,0) sin (9.85,-2) cos (10.1,0) -- (10.2,0);
    \draw[thick, black] (10.2,0) sin (10.45,1) cos (10.7,0) sin (10.95,-1) cos (11.2,0) -- (11.3,0);
    \draw[thick, black] (11.3,0) sin (11.55,2) cos (11.8,0) sin (12.05,-2) cos (12.3,0);

	\draw (0.5,2) node [anchor=south] {$a$};
	\draw (1.6,2) node [anchor=south] {$b$};
	\draw (2.7,2) node [anchor=south] {$a$};
    
	\draw (3.9,2) node [anchor=south] {$a$};
	\draw (5.0,2) node [anchor=south] {$b$};
	\draw (6.1,2) node [anchor=south] {$a$};
    
	\draw (7.3,2) node [anchor=south] {$a$};
	\draw (8.4,2) node [anchor=south] {$b$};
	\draw (9.5,2) node [anchor=south] {$a$};
    
	\draw (10.7,2) node [anchor=south] {$b$};
	\draw (11.8,2) node [anchor=south] {$a$};

\end{tikzpicture}}
\end{minipage}\hfill
\begin{minipage}[c][3cm][c]{\dimexpr0.5\textwidth-5pt\relax}
\resizebox{\textwidth}{!}{
\begin{tikzpicture}
	\draw[thick, black] (-1,0) -- (0,0);
    \draw[thick, black] (12.3,0) -- (13.3,0);

    \draw[thick, black] (0,0) sin (0.25,2) cos (0.5,0) sin (0.75,-2) cos (1,0) -- (1.1,0);
    \draw[thick, black] (1.1,0) sin (1.35,1) cos (1.6,0) sin (1.85,-1) cos (2.1,0) -- (2.2,0);
    \draw[thick, black] (2.2,0) sin (2.45,2) cos (2.7,0) sin (2.95,-2) cos (3.2,0) -- (3.4,0);
    
    \draw[thick, black] (3.4,0) sin (3.65,2) cos (3.9,0) sin (4.15,-2) cos (4.4,0) -- (4.5,0);
    \draw[thick, black] (4.5,0) sin (4.75,1) cos (5.0,0) sin (5.25,-1) cos (5.5,0) -- (5.6,0);
    \draw[thick, black] (5.6,0) sin (5.85,2) cos (6.1,0) sin (6.35,-2) cos (6.6,0) -- (6.8,0);
    
    \draw[thick, black] (6.8,0) sin (7.05,2) cos (7.3,0) sin (7.55,-2) cos (7.8,0) -- (7.9,0);
    \draw[thick, black] (7.9,0) sin (8.15,1) cos (8.4,0) sin (8.65,-1) cos (8.9,0) -- (9.0,0);
    \draw[thick, black] (9.0,0) sin (9.3,3.9) cos (9.55,0) sin (9.8,-3.9) cos (10.1,0) -- (10.2,0);
    \draw[thick, dashed, blue] (9.0,0) sin (9.25,2) cos (9.5,0) sin (9.75,-2) cos (10.0,0) -- (10.2,0);
    
    \draw[thick, dashed, red] (9.0,0) -- (9.1,0) sin (9.35,2) cos (9.6,0) sin (9.85,-2) cos (10.1,0) -- (10.2,0);
    \draw[thick, black] (10.2,0) sin (10.45,1) cos (10.7,0) sin (10.95,-1) cos (11.2,0) -- (11.3,0);
    \draw[thick, black] (11.3,0) sin (11.55,2) cos (11.8,0) sin (12.05,-2) cos (12.3,0);

	\draw (0.5,2) node [anchor=south] {$a$};
	\draw (1.6,2) node [anchor=south] {$b$};
	\draw (2.7,2) node [anchor=south] {$a$};
    
	\draw (3.9,2) node [anchor=south] {$a$};
	\draw (5.0,2) node [anchor=south] {$b$};
	\draw (6.1,2) node [anchor=south] {$a$};
    
	\draw (7.3,2) node [anchor=south] {$a$};
	\draw (8.4,2) node [anchor=south] {$b$};
	\draw (9.5,4) node [anchor=south] {$c$};
    
	\draw (10.7,2) node [anchor=south] {$b$};
	\draw (11.8,2) node [anchor=south] {$a$};

\end{tikzpicture}}
\end{minipage}

\begin{minipage}[c][4em][t]{\dimexpr0.5\textwidth-5pt\relax}
\captionof{figure}{The string $aba$ sent repeatedly over a channel as an ASK signal, with a desynchronization moment}
\end{minipage}\hfill
\begin{minipage}[c][4em][t]{\dimexpr0.5\textwidth-5pt\relax}
\captionof{figure}{The string $aba$ sent repeatedly over a channel as an ASK signal, with an echo}
\end{minipage}\hfill
\end{minipage}
\end{figure}
For example, in signal processing we may confidently rely upon periodicity, since we induce it ourselves and have an environment upon which we may make some assumptions. However, when trying to decode information which was not encoded by us, we may not expect to find periodicity. Even when the information was imbued with periodicity, if the environment exerts a degrading force, a posteriori it is entirely possible that it is no longer be periodic. If however it is not too degraded, it still holds faithful to its original form and hence exhibit quasi-periodicity. Note that the incurred perturbations may be inevitable in the typical usage environment, especially for industrial uses~\cite{georgiev1993digital}.



\paragraph*{\sc Related work.}

Quasi-periodicity was introduced by Ehrenfeucht in 1990 (according to ~\cite{ApostolicoB97}) in a Tech Report for Purdue University, even though in was not published in Elsevier until 1993~\cite{apostolico1993efficient}. Apostolico, Farach and Iliopoulos were the first to consider quasi-periodicity in computer science~\cite{ApostolicoFI91}. They define the quasi-period of a string to be the length of its shortest cover and present a linear (time and space) algorithm for computing it~\cite{ApostolicoFI91}. This notion attracted the attention of numerous researchers ~\cite{breslauer1992line,breslauer1994testing,li2002computing,moore1994optimal,moore1995correction}. The following surveys summarize the first decade of results:~\cite{apostolico1997periods,kociumaka2015fast,kolpakov2003finding}.

However, quasi-periodicity takes many forms, depending on the type of patterns we want to recover. Further work has been concerned with different variants such as seeds~\cite{guo2006computing}, the maximum quasi-periodic substring~\cite{pedersen2000finding}, k-covers~\cite{cole2005complexity}, $\lambda$-covers~\cite{guo2006computing}, enhanced covers~\cite{flouri2013enhanced}, partial covers~\cite{kociumaka2015fast}. Another variation point is the context, e.g. indeterminate strings~\cite{antoniou2008conservative} or weighted sequences~\cite{christodoulakis2006computation}. Some of the related problems are $\mathcal{NP}$-hard. 

For some applications, such as molecular biology and computer-assisted musical analysis, we need a weaker definition of quasi-periodicity. Thus, quasi-periodicity takes the form of approximate repetitions. We may define an approximatively repeating pattern as a substring whose occurrences leave very few gaps, or that all repetitions are near an ``original'' source. Landau and Schmidt study first this form of quasi-periodicity and focus on approximate tandem repeats~\cite{landau2001algorithm}.

In this paper we elaborate on the work of Amir et al. ~\cite{AmirLLLP17,AmirLLP17} who introduce \emph{approximate string covers}. 

Let $w$ be a string over the alphabet $\Sigma$. We say that $w$ is periodic if it is a succession of repetitions of some proper substring $p$ of it that do not overlap i.e. $w = p^n$, for some $n\in\mathbb{N}^*$. Note that for a given $w$ there may be multiple candidates. For example, $abaabaabaaba$ can be written as both $\left(abaaba\right)^2$ or $\left(aba\right)^4$. The period of a string $w$ is the shortest candidate string $p$. For instance, the period of $abaabaabaaba$ is $aba$. 


Let $w$ be a string over the alphabet $\Sigma$. We call $p$ a cover of $w$ if $p$ is shorter than $w$ and any character of $w$ belongs to some occurrence of $p$ in $w$. Equivalently, $w$ is covered by $p$ if $w$ is a succession of repetitions of $p$ that may or may not overlap. Note that a periodic string is always covered by its cover and any multiple of it and hence a string may admit multiple covers. As is the case for periods, we are only interested in the shortest cover. For instance the shortest cover of $abaabaabaaba$ is $aba$. 


Determining the shortest cover of a given string $w$ is called the Minimal \textbf{S}tring \textbf{C}over \textbf{P}roblem (\textbf{SCP} for short) and is solvable in linear time~\cite{ApostolicoFI91}. 

 Let $w$ be a string over the alphabet $\Sigma$. We call $p$ an \emph{approximate cover} of $w$, if $p$ is a cover of an ``approximation'' $w^\prime$ of $w$. The approximation error is the distance between $w$ and $w^\prime$ with respect to some metric. By abuse of notation we say that $p$ is the approximate string cover of $w$ if it is the shortest cover of the closest approximation $w^\prime$ of $w$ that admits a cover. Note that if $w$ admits a cover then its approximate string cover is its own shortest cover and the approximation is zero with regard to any metric. For example the approximate cover of $abaabaababa$ is $aba$.


Determining the approximate cover of a given string $w$ is called the \textbf{A}pproximate String \textbf{C}over \textbf{P}roblem (\textbf{ACP} for short). Amir et al. prove that \textbf{ACP} is NP-hard with respect to the Hamming distance~\cite{AmirLLP17}. 

Let $w$ be a string over the alphabet $\Sigma$. We call $p$ a seed of $w$ if $\lvert p\rvert < \lvert w\rvert$ and there exists a super-string $w^\prime$ of $w$ such that $p$ is a cover of $w^\prime$. When the error tolerance is small, with a small degree of incertitude we can find in polynomial time\cite{AmirLLLP17} a small set of candidates containing either the approximate cover of $w$, $p$, or a seed of $p$.

\paragraph*{\sc Our results}
In this paper we follow up on the work of Amir et al. ~\cite{AmirLLP17,AmirLLLP17} and investigate the \textbf{ACP}. 
In Section~\ref{sec:preliminaries} we introduce notation and we define formally the \textbf{ACP} problem. In Section~\ref{sec:approx} we present a polynomial approximation algorithm for \textbf{ACP} that returns an approximate cover that matches at least $\Omega\left(\sqrt{n}\right)$ characters of the given input $w$, where $n$ is the number of characters matched by the (best) approximate cover of $w$. Then, in Section~\ref{sec:fpt} we design a fixed-parameter (probabilisitic) algorithm for solving the \textbf{ACP} for (super)addtivie metrics---the (super)aditive metrics are also defined in Section~\ref{sec:fpt}.  

In Section~\ref{sec:pseudo} we show that \textbf{ACP} naturally extends to pseudometrics and that for a family of pseudometrics, which we call \emph{homogenous additive} the complexity of \textbf{ACP} remains unchanged. Finally, in Section~\ref{sec:block} we show that the Hamming distance over any unbounded alphabet is equivalent to an extended metric on any fixed size alphabet. We call this metric a \emph{block variation} of the Hamming distance. We prove that \textbf{ACP} is just as hard with regard to an additive (extended) (pseudo)metric as it is with regard to its block variation. Due to space constraints some proofs are placed in the appendix.

\section{Preliminaries}
\label{sec:preliminaries}

For \({m\leq n}\in\mathbb{N}\), let \(\overline{{m,\,n}}=\lbrace {m,\,m+1,\,\dots,\,n}\rbrace\).
For two symbols \({x,\,y\in X}\) let \(\delta_{x,\,y}\) be the Kronecker delta, i.e. $\delta_{x,\,y}=1$ if ${x=y}$ or $0$ if ${x\neq y}$. For a string $w$ and a character $c \in Sigma$, let $freq_w(c)$ be the number of occurences of $c$ in $w$.

\begin{definition}[tilings]
We define a tiling of size n to be a pair \(\left(\mathcal{I},\,i\right)\) where \(\mathcal{I}\subseteq\overline{1,\,{n}}\) and \(i:\overline{1,\,\lvert\mathcal{I}\rvert}\rightarrow\mathcal{I}\) such that:
\begin{itemize}
\item \(\forall {i}\in\overline{1,\,\lvert\mathcal{I}\rvert}\ \mathcal{I}_{i}= i\left({i}\right)\)
\item \(\mathcal{I}_1=1,\,\mathcal{I}_{last}=\mathcal{I}_{\lvert\mathcal{I}\rvert},\,\lVert\mathcal{I}\rVert = n+1-\mathcal{I}_{last}\)
\item \(\forall {i}\in\overline{1,\,\lvert\mathcal{I}\rvert -1}\ 0\leq\mathcal{I}_{i+1}-\mathcal{I}_{i}\leq\lVert\mathcal{I}\rVert\)
\end{itemize}

\end{definition}

We consider \(\mathcal{T}_{n}\) to be the set of size-n tilings and \(i\) will generally be omitted, being deferred to the subscript notation above, where additionally we write \(\mathcal{I}_\textmd{last}=\mathcal{I}_{\lvert\mathcal{I}\rvert}\). 

\begin{example}
For instance \(\mathcal{I}=\lbrace1,\,4,\,6\rbrace\in\mathcal{T}_8\) is a tiling with \(\lVert\mathcal{I}\rVert=3\). For an illustration see Figure~\ref{fig:example_tiling}.
\end{example}

\begin{figure}
\begin{center}\begin{tikzpicture}
\draw[step=1cm,gray,very thin] (0,1) grid (8,2);
\foreach \x in {1,2,3,4,5,6,7}
	\node [anchor=south] at (\x cm-0.5cm,2) {$\x$};
\node [anchor=south] at (8cm, 2) {${n}=8$};
\draw [thick, dashed, gray] (0 cm, 1) rectangle (3cm, 2);
\draw [thick, dashed, gray] (3 cm, 1) rectangle (6cm, 2);
\draw [thick, dashed] (5 cm, 1) rectangle (8cm, 2);
\draw [thick, red] (5 cm,1) -- (8 cm,1);
\node [anchor=south, color=red] at (6.5cm, 1) {$\lVert\mathcal{I}\rVert=3$};
\node [anchor=north] at (0.5, 1) {$\mathcal{I}_1$};
\node [anchor=north] at (3.5, 1) {$\mathcal{I}_2$};
\node [anchor=north] at (5.5, 1) {$\mathcal{I}_\textmd{last}$};
\end{tikzpicture}\end{center}
\caption{Example of a tiling}
\label{fig:example_tiling}
\end{figure}
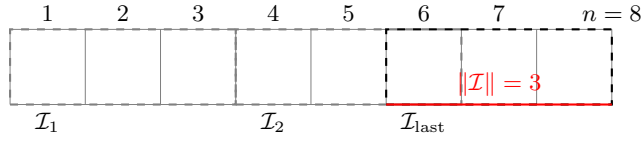

\begin{definition} [valid tilings]
Let \(\mathcal{I}\in\mathcal{T}_n\) be a tiling. We say that a word \(w\in\Sigma^*\) over a given alphabet leads to a valid tiling \(\left(w,\,\mathcal{I}\right)\) iff \(\lvert w\rvert=\lVert\mathcal{I}\rVert\) and \(\forall i\in\overline{1,\,\lvert\mathcal{I}\rvert-1},\,j \in\overline{1,\,\mathcal{I}_i-\mathcal{I}_{i+1}+\lVert\mathcal{I}\rVert-1}\ {w}_{j}={w}_{\mathcal{I}_{i+1}-\mathcal{I}_{i}+{j}}\) i.e. there are no conflicts in the tiling itself. In this case we say that w is a cover of \(\mathcal{I}\left({w}\right)\in\Sigma^{n}\) where \(\forall {i}\in\overline{1,\,\lvert\mathcal{I}\rvert},\, {j}\in\overline{1,\,\lVert\mathcal{I}\rVert}\ \mathcal{I}\left( {w}\right)_{\mathcal{I}_{i}-1+{j}}={w}_{j}\).
\end{definition}

\begin{example}
For the tiling in Figure~\ref{fig:example_tiling} we want to find a word \(w\in\Sigma^3\) such that \(w_1=w_3\).

For instance, if \(w=aba\) we obtain \(\mathcal{I}\left(w\right)=abaababa\), but \({w}^\prime=abb\) does not lead to a valid tiling. For an illustration see Figure~\ref{fig:valid-tiling}.
\end{example}

\vspace*{-0.5cm}
\begin{figure}
\begin{center}\begin{tikzpicture}
\draw[step=1cm,gray,very thin] (0,1) grid (8,2);
\foreach \x in {1,2,3,4,5,6,7}
	\node [anchor=south] at (\x cm-0.5cm,2) {$\x$};
\node [anchor=south] at (8cm, 2) {${n}=8$};
\draw [thick, dashed, gray] (0 cm, 1) rectangle (3cm, 2);
\draw [thick, dashed, gray] (3 cm, 1) rectangle (6cm, 2);
\draw [thick, dashed] (5 cm, 1) rectangle (8cm, 2);
\draw [thick, red] (5 cm,1) -- (8 cm,1);
\node [anchor=north] at (0.5, 1) {$\mathcal{I}_1$};
\node [anchor=north] at (3.5, 1) {$\mathcal{I}_2$};
\node [anchor=north] at (5.5, 1) {$\mathcal{I}_\textmd{last}$};
\node [color=blue] at (0.5, 1.5) {a};
\node [color=blue] at (1.5, 1.5) {b};
\node [color=blue] at (2.5, 1.5) {a};
\node [color=red] at (3.5, 1.5) {a};
\node [color=red] at (4.5, 1.5) {b};
\node [color=red] at (5.45, 1.55) {a};
\node [color=green] at (5.55, 1.45) {a};
\node [color=green] at (6.5, 1.5) {b};
\node [color=green] at (7.5, 1.5) {a};
\draw [->, thick, color=green] (6.5, 0.5) -- (5.7, 1.2);
\node [color=green, thick, anchor=west] at (6.5, 0.5) {match};
\end{tikzpicture}\end{center}

\begin{center}\begin{tikzpicture}
\draw[step=1cm,gray,very thin] (0,1) grid (8,2);
\foreach \x in {1,2,3,4,5,6,7}
	\node [anchor=south] at (\x cm-0.5cm,2) {$\x$};
\node [anchor=south] at (8cm, 2) {${n}=8$};
\draw [thick, dashed, gray] (0 cm, 1) rectangle (3cm, 2);
\draw [thick, dashed, gray] (3 cm, 1) rectangle (6cm, 2);
\draw [thick, dashed] (5 cm, 1) rectangle (8cm, 2);
\draw [thick, red] (5 cm,1) -- (8 cm,1);
\node [anchor=north] at (0.5, 1) {$\mathcal{I}_1$};
\node [anchor=north] at (3.5, 1) {$\mathcal{I}_2$};
\node [anchor=north] at (5.5, 1) {$\mathcal{I}_\textmd{last}$};
\node [color=blue] at (0.5, 1.5) {a};
\node [color=blue] at (1.5, 1.5) {b};
\node [color=blue] at (2.5, 1.5) {b};
\node [color=red] at (3.5, 1.5) {a};
\node [color=red] at (4.5, 1.5) {b};
\node [color=red] at (5.45, 1.55) {b};
\node [color=green] at (5.55, 1.45) {a};
\node [color=green] at (6.5, 1.5) {b};
\node [color=green] at (7.5, 1.5) {b};
\draw [->, thick, color=red] (6.5, 0.5) -- (5.7, 1.2);
\node [color=red, thick, anchor=west] at (6.5, 0.5) {conflict};
\end{tikzpicture}\end{center}

\caption{Example of a valid and an invalid tiling}
\label{fig:valid-tiling}
\end{figure}
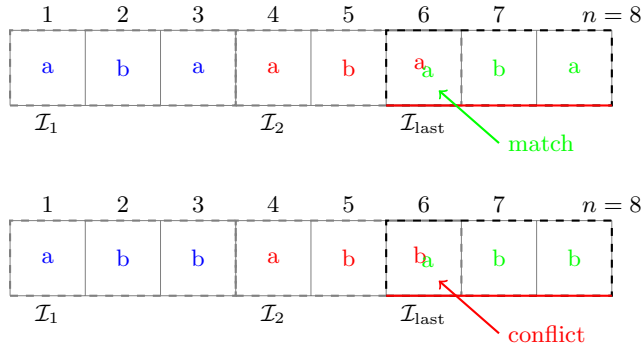
\vspace*{-0.5cm}
We are interested in the opposite process: given a string \(w\in\Sigma^*\) we are to determine a cover of it. This must be a substring, and moreover if there be multiple covers we aim for the smallest one, because $w$ is always a cover of $w$. This is called the minimal cover problem. If s is the minimal cover of $w$, then $s$ is the minimal cover of $s$. This stems from the fact that a cover of a cover of a string is also a cover of that string.

\begin{problem} [Minimal \textbf{S}tring \textbf{C}over \textbf{P}roblem, SCP]
\label{prob:scp}
Given \({w}\in\Sigma^*\) find \(\arg\,\min\lbrace\lvert{c}\rvert\vert{c}\in\Sigma^*,\,\exists\mathcal{I}\in\mathcal{T}_{\lvert{w}\rvert},\,{w}=\mathcal{I}\left({c}\right)\rbrace\) where any such \({c}\) leads to a valid tiling \(\left({c},\,\mathcal{I}\right)\).
\end{problem}

It is clear that we are looking for a \(\mathcal{I}\in\mathcal{T}_{\lvert{w}\rvert}\) and moreover we request that \(\lVert\mathcal{I}\rVert<\lvert{w}\rvert\) so as not to consider the trivial solution \(\lbrace1\rbrace\in\mathcal{T}_{\lvert{w}\rvert}\) and we would like to relax this problem by means of approximation. 

\begin{problem} [\textbf{A}pproximate String \textbf{C}over \textbf{P}roblem, ACP - specific version]
\label{prob:acp}
Let \(\delta\) be a metric over \(\Sigma^*\) and \(\Sigma_\mathcal{I}\subseteq\Sigma^{\lVert\mathcal{I}\rVert}\) denote the set of words that lead to valid tilings over it. We define \({w}^*=\mathcal{I}^*\left({s^*}\right)=\arg\,\min\big\lbrace\lvert{c}\rvert\vert{c}\in\arg\,\min\ \big\lbrace\delta\left({w},\,\mathcal{I}\left({s}\right)\right)\big\vert\mathcal{I}\in\mathcal{T}_{\lvert{w}\rvert},\,\lVert\mathcal{I}\rVert={m}<\lvert{w}\rvert,\,{s}\in\Sigma_\mathcal{I}\big\rbrace\big\rbrace\). The goal of the problem is to find $s^*,\,\mathcal{I}^*$. Informally, $s^*$ is the string of fixed size m, that produces valid tiling \(\mathcal{I}^*\left({s}^*\right)\) of minimum distance to $w$.
\end{problem}

\begin{problem} [\textbf{A}pproximate String \textbf{C}over \textbf{P}roblem, ACP - general version]
\label{prob:acpg}
Let \(\delta\) be a metric over \(\Sigma^*\) and \(\Sigma_\mathcal{I}\subseteq\Sigma^{\lVert\mathcal{I}\rVert}\) denote the set of words that lead to valid tilings over it. We define \({w}^*=\mathcal{I}^*\left({s^*}\right)=\arg\,\min\big\lbrace\lvert{c}\rvert\vert{c}\in\arg\,\min\ \big\lbrace\delta\left({w},\,\mathcal{I}\left({s}\right)\right)\big\vert\mathcal{I}\in\mathcal{T}_{\lvert{w}\rvert},\,\lVert\mathcal{I}\rVert<\lvert{w}\rvert,\,{s}\in\Sigma_\mathcal{I}\big\rbrace\big\rbrace\). The goal of the problem is to find $s^*,\,\mathcal{I}^*$. Informally, $s^*$ is the shortest string that produces valid tiling \(\mathcal{I}^*\left({s}^*\right)\) of minimum distance to $w$.
\end{problem}



\section{A polynomial-time approximation algorithm}
\label{sec:approx}
Consider Problem~\ref{prob:acp} (ACP - specific version) for which we are asked for a fixed-size string \({s}^*\) such that \({m}=\lvert{s}^*\rvert\) that produces a valid tiling \({w}^*=\mathcal{I}^*\left({s}^*\right)\) of minimum distance to the input string, w. This section addresses approximation algorithms for the case of the Hamming distance.

\begin{lemma} The string \({s}^\prime=\alpha^{m}\) always produces a valid tiling \({w}^\prime=\alpha^{n}\) and \(\lvert{w}\rvert-{d}\left({w},\,{w}^\prime\right)={freq}_{w}\left(\alpha\right)\)
\end{lemma}
\begin{proof}
Consider the tiling \(\mathcal{I}\) given by \(\overline{1,\,{n-m+1}}\ni {i}\rightarrow i\left({i}\right)={i}\). We have that \(\forall{i}\in\overline{1,\,{m-1}}\ {s}^\prime_{i}=\alpha={s}^\prime_{i+1}={s}^\prime_{m-\left(m-1\right)+i}\) and hence we indeed have a valid tiling \({w}^\prime=\mathcal{I}\left({s}^\prime\right)=\alpha^{n}\) for which \(\lvert{w}\rvert-{d}\left({w},\,{w}^\prime\right)=\sum\delta_{{w}_{i},\,{w}^\prime_{i}}=\sum\delta_{{w}_{i},\,\alpha}={freq}_{w}\left(\alpha\right)\)
\qed \end{proof}
\begin{corollary} We can always match at least the most frequent character i.e. \(\lvert{w}\rvert-{d}\left({w},\,{w}^*\right)\geq \underset{\alpha}{\max}\,{freq}\left(\alpha\right)={freq}_{max}\)
\end{corollary}
\begin{definition}[Cover Efficiency]
We define the cover efficiency function as:
\begin{center}\(\Sigma^*\times\Sigma^*\ni\left({w},\,{w}^\prime\right)\rightarrow\eta\left({w},\,{w}^\prime\right)=\frac{\lvert{w}\rvert-{d}\left({w},\,{w}^\prime\right)}{\lvert{w}\rvert-{d}\left({w},\,{w}^*\right)}\in\left[0,\,1\right]\)
\end{center}
Thus, an algorithm for the ACP problem 
is an \(\mathcal{O}\left({f}\left(\lvert{w}\rvert\right)\right)\) approximation of  iff \(\frac{1}{\eta_\mathcal{A}}\in\mathcal{O}\left({f}\left(\lvert{w}\rvert\right)\right)\), where $\eta_\mathcal{A}$ is the efficiency function of the algorithm. 
\end{definition}

\begin{lemma}If \({freq}_{max}\in\Omega\left(\sqrt{\lvert{w}\rvert}\right)\) then the algorithm providing \(\overline{1,\,{n-m+1}}\ni {i}\rightarrow i\left({i}\right)={i}\) and \(\left(\underset{\alpha}{\arg\max}\ {freq}_{w}\left(\alpha\right)\right)^{m}\in\Sigma_\mathcal{I}\) is a \(\Omega\left(\sqrt{\lvert{w}\rvert}\right)\) approximation.
\end{lemma}
\begin{proof}
\(\frac{1}{\eta}=\frac{\lvert{w}\rvert-{d}\left({w},\,{w}^*\right)}{\lvert{w}\rvert-{d}\left({w},\,{w}^\prime\right)}\leq\frac{\lvert{w}\rvert}{{freq}_{max}}\in\mathcal{O}\left(\frac{\lvert{w}\rvert}{\sqrt{\lvert{w}\rvert}}\right)=\mathcal{O}\left(\sqrt{\lvert{w}\rvert}\right)\)
\qed \end{proof}

\begin{lemma}\(\lvert{w}\rvert-{d}\left({w},\,{w}^*\right)\leq{m}\cdot{freq}_{max}\)
\end{lemma}
\begin{proof}\(\lvert{w}\rvert-{d}\left({w},\,{w}^*\right)=\sum_{i}\delta_{{w}_{i},\,{w}^*_{i}}=\sum_{{i},\,\alpha}\delta_{{w}_{i},\,\alpha}\delta_{\alpha,\,{w}^*_{i}}\leq\sum_{{i},\,{j}}\delta_{{w}_{i},\,{s}^*_{j}}\delta_{{s}^*_{j},\,{w}^*_{i}}\leq
\sum_{{i},\,{j}}\delta_{{w}_{i},\,{s}^*_{j}}=\sum_{{j}}{freq}_{{w}}\left({s}^*_{j}\right)\leq\sum_{{j}}{freq}_{max}={m}\cdot{freq}_{max}\)
\qed \end{proof}
\begin{corollary}If \({m}\in\mathcal{O}\left(\sqrt{\lvert{w}\rvert}\right)\) then the algorithm providing \(\overline{1,\,{n-m+1}}\ni {i}\rightarrow i\left({i}\right)={i}\) and \(\left(\underset{\alpha}{\arg\max}\ {freq}_{w}\left(\alpha\right)\right)^{m}\in\Sigma_\mathcal{I}\) is a \(\Omega\left(\sqrt{\lvert{w}\rvert}\right)\) approximation.
\end{corollary}
\begin{proof}
\(\frac{1}{\eta}=\frac{\lvert{w}\rvert-{d}\left({w},\,{w}^*\right)}{\lvert{w}\rvert-{d}\left({w},\,{w}^\prime\right)}\leq\frac{{m}\cdot{freq}_{max}}{{freq}_{max}}={m}\in\mathcal{O}\left(\sqrt{\lvert{w}\rvert}\right)\)
\qed \end{proof}

\begin{lemma}
If \({m}\in\Omega\left(\sqrt{\lvert{w}\rvert}\right),\,{m}\leq\lfloor\lvert{w}\rvert/3\rfloor\) and \(\alpha\in\Sigma\) then there is a valid tiling \(\mathcal{I}\in\mathcal{T}_{\lvert{w}\rvert}\) with \({s}^\prime=\overline{\alpha\dots\alpha{w}_{\lceil{m}/3\rceil+1}{w}_{\lceil{m}/3\rceil+2}\dots{w}_{{m}-\lceil{m}/3\rceil-1}\alpha\dots\alpha}\in\Sigma_\mathcal{I}\).
\end{lemma}
\begin{proof}
Let \(\mathcal{T}_{n}^{k}=\lbrace\mathcal{I}\in\mathcal{T}_{n}\vert\lvert\mathcal{I}\rvert={k},\,{s}^\prime\in\Sigma_\mathcal{I}\rbrace\) and \(\mathcal{T}^{k}=\underset{{n}}{\cup}\mathcal{T}_{n}^{k}\). We are interested in the quantity \(\lvert\mathcal{I}\left({s}^\prime\right)\rvert\) which induces \(\mathcal{N}^{k}=\lbrace{n}\in\mathbb{N}\vert\exists\mathcal{I}\in\mathcal{T}_{n},\,{s}^\prime\in\Sigma_\mathcal{I}\rbrace\). 

Note that \({s}^\prime\in\Sigma_\mathcal{I}\Leftrightarrow\lVert\mathcal{I}\rVert={m},\,\underset{{i}}{\min}\ \mathcal{I}_{i+1}-\mathcal{I}_{i}\geq{m}-\lceil{m}/3\rceil\). Hence, \(\mathcal{N}^{k+1}=\lbrace{n}+{d}\vert{n}\in\mathcal{N}^{k},\,{d}\in\overline{{m}-\lceil{m}/3\rceil,\,{m}}\rbrace\).

By definition, \(\mathcal{N}^1=\lbrace{m}\rbrace\). Hence \(\mathcal{N}^2=\overline{2{m}-\lceil{m}/3\rceil,\,2{m}}\) and generally \(\mathcal{N}^{k}=\overline{{k\cdot m}-\left({k}-1\right)\lceil{m}/3\rceil,\,{k\cdot m}}\). If there is a \({k}\) such that \({k\cdot m}\geq{\left(k+1\right)\cdot m}-{k}\lceil{m}/3\rceil\Leftrightarrow {k}\geq\frac{{m}}{\lceil{m}/3\rceil}\) then if \(\lvert{w}\rvert\geq{m\cdot k}\) we have \(\lvert{w}\rvert\in\underset{{k}^\prime}{\cup}\mathcal{N}^{{k}^\prime}\) and the theorem is proven. Since \({k}= 3\) is a viable option and \(\lvert{w}\rvert\geq 3\lfloor\lvert{w}\rvert/3\rfloor\geq 3{m}\), then this is the case.
\qed \end{proof}

Combining the previous results we obtain the desired approximation ratio.

\begin{theorem}
Let \(\alpha\) be the most frequent character in \({w}\in\Sigma^*\). We can compute an \(\mathcal{O}\left(\sqrt{\lvert{w}\rvert}\right)\) approximation for the \textbf{ACP} problem in \(\mathcal{O}\left(\lvert{w}\rvert^3\right)\) time.
\end{theorem}
\begin{proof}
Let \(\alpha=\underset{\beta}{\arg\max}\ {freq}_{w}\left(\beta\right)\) which we can obtain in linear time. The algorithm works as follows.

\begin{enumerate}

\item If \({m}\in\mathcal{O}\left(\sqrt{\lvert{w}\rvert}\right)\) then return \(\alpha^{m}\).

\item If \({m}\leq\lceil{n}/3\rceil\) then return \(\overline{\alpha\dots\alpha{w}_{\lceil{m}/3\rceil+1}{w}_{\lceil{m}/3\rceil+2}\dots{w}_{{m}-\lceil{m}/3\rceil-1}\alpha\dots\alpha}\).

\item Otherwise, we  have \(\lvert\mathcal{I}\rvert\leq 3\) and thus we can check all the possibilities in cubic time.
\end{enumerate}
\qed \end{proof}

\section{An FPT Algorithm for the Approximate String Cover Problem}
\label{sec:fpt}
\begin{definition}[Product Metrics]
Let \(\lbrace\left({X}_{i},\,{d}_{i}\right)\rbrace_{{i}\in\overline{1,\,{n}}}\) be metric spaces. Then a metric \({d}\) such that \(\left(\underset{{i}=1}{\overset{{n}}{\prod}}{X}_{i}\right)^2\ni\left(\underset{{i}=1}{\overset{{n}}{\oplus}}\mathbf{x}_{i},\,\underset{{i}=1}{\overset{{n}}{\oplus}}\mathbf{y}_{i}\right)=\left(\mathbf{x},\,\mathbf{y}\right)\rightarrow{d}\left(\mathbf{x},\,\mathbf{y}\right)\geq\underset{{i}=1}{\overset{{n}}{\sum}}{d}_{i}\left(\mathbf{x}_{i},\,\mathbf{y}_{i}\right)\in\mathbb{R}_+\) is called superadditive. We define similarly subadditive and additive metrics.
\end{definition}

The Hamming distance over \(\Sigma^{n}\) is the additive metric for \({X}_{i}=\Sigma\) and \({d}_{i}\left({x},\,{y}\right)=1-\delta_{{x},\,{y}}\) with \({i}\in\overline{1,\,{n}}\).  Another additive metric is that for shift spaces i.e. \({d\left(u,\,v\right)}=\sum_{i}\frac{{d_i}\left({u},\,{v}\right)}{2^{i}}\).


\begin{theorem}
For a (super)additive metric ACP can be solved in \(\mathcal{O}\left(\lvert\Sigma\rvert^{m}{m}^2\lvert{w}\rvert^2\right)\) (probabilistic) time with \(\mathcal{O}\left(\lvert{w}\rvert\right)\) space.
\end{theorem}
\begin{proof}
For a given vector \(\mathbf{v}\in\mathbb{Z}^{\lvert{w}\rvert+1-\lvert{s}\rvert}\) such that \(\mathbf{v}_1=1\) and \(\forall{i}\in\overline{2,\,\lvert{w}\rvert+1-\lvert{s}\rvert}\ \mathbf{v}_{i}> 0\) iff there exists a valid tiling of arbitrary length \(\mathcal{I}\) such that \(\mathcal{I}_{last}={i}\), \({s}\in\Sigma_\mathcal{I}\) and then \(\mathbf{v}_{\mathcal{I}_{last}}=\mathcal{I}_{\lvert\mathcal{I}\rvert-1}\), we say that \(\mathbf{v}\) encodes \(\mathcal{J}\) i.e. \(\mathcal{J}\in{Dec}\left(\mathbf{v}\right)\) iff \(\forall{i}\in\overline{1,\,\lvert\mathcal{J}-1\rvert}\ \mathbf{v}_{\mathcal{J}_{i+1}}=\mathcal{J}_{i}\). 

Given the last position of the encoded tiling we can always recover it by following the backward orbit described above i.e. there is a function \(\mathbf{v},\,{i}\rightarrow{dec}\left(\mathbf{v},\,{i}\right)\in\mathcal{T}_{i-1+\lvert{s}\rvert}\). Consider the sequence \(\mathit{i}_1= \mathbf{v}_{i},\,\mathit{i}_{j+1}=\mathbf{v}_{\mathit{i}_{j}}\) which is  cyclic after a point with a 1-cycle around 1. Hence \(\mathit{i}^{-1}\left(1\right)=\lbrace 1,\,{k}\neq 1\rbrace\) and so we define \(\lvert{dec}\left(\mathbf{v},\,{i}\right)\rvert={k}+1,\,{dec}\left(\mathbf{v},\,{i}\right)_{j}=\mathit{i}_{k+2-j}\). This applies if and only if the first position is not 1 in which case we have the trivial tiling \({dec}\left(\mathbf{v},\,1\right)=\lbrace 1\rbrace\in\mathcal{T}_{\lvert s\rvert}\)

If \({d}\left({w},\,\mathcal{I}\left({s}\right)\right)<\delta\) and \(\mathcal{J}=\left(\mathcal{I}_1,\,\dots,\,\mathcal{I}_{\lvert\mathcal{I}\rvert-1}\right)\) then due to supperadditivity \({d}\left(\overline{{w}},\,\mathcal{J}\left({s}\right)\right)<\delta\) where \(\overline{{w}}\) is the appropiate truncation of w and moreover if \(\mathbf{v}\) encodes \(\mathcal{I}\) then it also encodes \(\mathcal{J}\). Hence, there exists \(\mathcal{I}\) under the tolerance limit iff it is encoded by a vector which only encodes tolerable tilings. This gives the FPT algorithm for super-additive metrics, using the sequence \(\mathbf{v}^{{n}\in\mathbb{N}^*}\subseteq\mathbb{Z}^{\lvert{w}\rvert+1-\lvert{s}\rvert}\), defined as follows: 
\begin{center}\(\begin{cases}
\mathbf{v}^1_{i}=\delta_{{i},\,1} \\
\mathbf{v}^{n+1}_{i}={choose-one}\left(\lbrace{j}\vert\mathbf{v}^{n}_{j}>0,\,\lvert{s}\rvert+{j}-{i}\in\Delta_{s},\,{d}\left(\overline{{w}},\,{dec}\left(\mathbf{v},\,{j}\right)\cup\lbrace{i}\rbrace\right)<\delta\rbrace\right)
\end{cases}\)\end{center}

Moreover, if the metric is additive, then a greedy algorithm works and as such we have the sequence \(\mathbf{v}^{{n}\in\mathbb{N}^*}\subseteq\mathbb{Z}^{\lvert{w}\rvert+1-\lvert{s}\rvert}\), defined as follows: \begin{center}\(\begin{cases}
\mathbf{v}^1_{i}=\delta_{{i},\,1} \\
\mathbf{v}^{n+1}_{i}=\underset{{j}}{\arg\min}\left(\lbrace{d}\left(\overline{{w}},\,{dec}\left(\mathbf{v},\,{j}\right)\cup\lbrace{i}\rbrace\right)\vert\mathbf{v}^{n}_{j}>0,\,\lvert{s}\rvert+{j}-{i}\in\Delta_{s}\rbrace\right)
\end{cases}\)\end{center}
To prove that the greedy approach let \(\mathcal{I}^*\) the optimum tiling for a given tile s. By construction, since a tiling ending with \(\mathcal{I}^*_{last}\) exists, \(\mathbf{v}_{\lvert{w}\rvert+1-\lvert{s}\rvert} > 0\) and hence our algorithm always provides a tiling. Moreover, it agrees with the optimum tiling on the last element. Hence, if \(\lvert{w}\rvert = \lvert{s}\rvert\) the optimum tiling is provided. 
\qed \end{proof}

\section{Pseudometrics and Halo Factorization}
\label{sec:pseudo}

Firstly, we generalize the \textbf{ACP} for pseudometric spaces. Recall the definition of a metric space.
\begin{definition}[Metric Spaces] Let $X$ be a set and $d:X\times X\rightarrow \mathbb{R}$ map the pairs of $points$ in $X$ to the reals. $d$ is a metric on $X$ if and only if:

$d\left(x,\,y\right)\geq 0\ \forall x,\,y\in X\ \textmd{(positivity)}$

$d\left(x,\,y\right)=d\left(y,\,x\right)\ \forall x,\,y\in X\ \textmd{(symmetry)}$

$d\left(x,\,y\right)+d\left(y,\,z\right)\geq d\left(x,\,z\right)\ \forall x,\,y,\,z\in X\ \textmd{(triangle inequality)}$

$x=y\Leftrightarrow d\left(x,\,y\right)=0\ \forall x,\,y\in X \textmd{(identity of indiscernibles)}$
\end{definition}

Let $\Sigma$ be the set of message types in a communication protocol where $q$ is some poll, with $y$ and $n$ being acceptable answers and $NACK$ meaning that the poll was not accepted by the other party. It might be that the pattern we are looking for contains a successfull dialogue, regardless of its content per se. Then, for us the strings $qy$ and $qn$ should be indiscernible. Metric spaces do not allow this. If we relax the identity of indiscernibles we obtain a pseudometric.

\begin{definition}[Pseudometric Spaces] Let $X$ be a set and $d:X\times X\rightarrow \mathbb{R}$ map the pairs of $points$ in $X$ to the reals. $d$ is a pseudometric on $X$ if and only if:

$d\left(x,\,y\right)\geq 0\ \forall x,\,y\in X\ \textmd{(positivity)}$

$d\left(x,\,y\right)=d\left(y,\,x\right)\ \forall x,\,y\in X\ \textmd{(symmetry)}$

$d\left(x,\,y\right)+d\left(y,\,z\right)\geq d\left(x,\,z\right)\ \forall x,\,y,\,z\in X\ \textmd{(triangle inequality)}$

$x=y\Rightarrow d\left(x,\,y\right)=0\ \forall x,\,y\in X$

\end{definition}




Let $\approx$ denote indiscernibility on $X$ with regard to the pseudometric $d$. Then $\approx$ is an equivalence relation on $X$ and $d$ is a metric on $X/\approx$. The elements of $X/\approx$ are called \emph{halos}, i.e. $\hat{x}=\lbrace y\in X\vert d\left(x,\,y\right)=0\rbrace$ is the halo around $x$. We prove that if $d$ is additive and homogenous \textbf{ACP} has the same complexity on $\left(X,\,d\right)$ as it does on $\left(X/\approx,\,d\right)$. 

This result is particularly useful when analyzing patterns in communication which can be recovered from use of metadata alone. Consider a validation proccedure \cite{SAS} in which a device sends a value to be transferred to a client's account by a host, the host replies with the same value and an asset number for the device (to confirm the athority of the host) and an index number to be written on the printed ticket and gets a print confirmation from the device. Any of these steps can go wrong, but a successfull transaction is easy to recognize. Two messages can be considered indiscernible if they have the same type. 

\begin{remark}
For two indiscernible strings, $w\approx w^\prime$, \textbf{ACP} has the same solution with respect to $w$ as it does with respect to $w^\prime$. Since $w \approx \hat{w}$ from any solution over the metric space we  obtain a solution over the pseudometric space, and all we need to do the oposite is to be able to quickly compute the factorization of a given string. This can only be done quickly under some additional conditions such as pseudometric additivity and homogeneity. 
\end{remark}


\begin{definition}[Additive Metrics] Let $\lbrace \left(X_1,\,d_2\right),\,\left(X_2,\,d_2\right)\,\dots,\,\left(X_n,\,d_n\right)\rbrace$ be a family of (pseudo)metric spaces. A (pseudo)metric $d$ over $\underset{i=1}{\overset{n}\oplus}X_i$ is additive if 
$ d\left(x_1x_2\dots x_n,\,y_1y_2\dots y_n\right)=\underset{i=1}{\overset{n}\sum}d_i\left(x_i,\,y_i\right),\,\forall x_i,\,y_i\in X_i$.
Moreover $d$ is homogenous if $d_i=d_j\forall i,\,j\in\lbrace{1,\,2,\,\dots,\,n}\rbrace$.
\end{definition}

\begin{remark}
If $d$ is an additive pseudometric over $\underset{i=1}{\overset{n}\oplus}X_i$ then $ x_1x_2\dots x_n \approx y_1y_2\dots y_n \Leftrightarrow x_i\approx y_i \forall i\in\lbrace 1,\,2,\,\dots,\,n\rbrace $.
\end{remark}

\begin{corollary}
A pseudometric $d$ is additive over $\underset{i=1}{\overset{n}\oplus}X_i$ if and only if it is additive over $\underset{i=1}{\overset{n}\oplus}X_i/\approx$.
\end{corollary}

If we have access to the function that mapped each $X_1$ to $X_1/\approx$, then factorizing the indiscernibles is an easy task, since we can perform it element by element. This map can be computed  in quadratic time, $\mathcal{O}\left(\lvert X\rvert\lvert X/\approx\rvert\right)$. However, we have to prove that collisions are handled properly.

For a given string $s$ to be a cover of some $w$, $w$ has to be a sequence of repetitions of $s$, with some eventual overlaps. On these overlaps, some suffix of $s$ has to match some prefix. In \cite{AmirLLLP17}, Amir et al.  represent this using \emph{string masks}, where the mask $m$ of a string $s$ is a vector $m\left[1,\,\dots,\,\lvert s\rvert\right]$ where $m_i=1$ if and only if the $i$-length prefix matches the $i$-length suffix of $s$. We choose to represent legal overlaps as sets.

\begin{definition}
For a given string $s$, let $\Delta_s$ be the set of legal overlaps of $s$:

$$\Delta_s = \lbrace i\in\lbrace1,\,\dots,\,\lvert s\rvert\rbrace \vert s_j = s_{\lvert s\rvert -i +j}\forall j\in\lbrace{1,\,\dots,\,i}\rbrace\rbrace$$
\end{definition}

In our setting, to factorize each character is equivalent to sequentially replace each occurence of a character $x$ with its chosen representative $\hat{x}$. This is where homogeneity comes into play: if $\hat{x}^i\neq \hat{x}^j$ for some $i,\,j\in\lbrace{1,\,\dots,\,\lvert s\rvert}$, then their collision is restricted, and thus the best tiling may be invalidated by factorization (consider the case where for some characters $x,\,y\in X_i\cap X_j$ we had $d_i\left(x,\,y\right)=0$ but $d_j\left(x,\,y\right)\neq 0$ ). 

\begin{lemma}Let \(\lbrace\left({X}_{i},\,{d}_{i}\right)\rbrace\) be pseudometric spaces and \({d}\) be the additive pseudometric. Then \(\hat{{u}}=\hat{{v}}\Leftrightarrow\forall{i}\in\overline{1,\,{n}}\ \hat{{u}_{i}}=\hat{{v}_{i}}\).
\end{lemma}
\begin{proof}
\(\hat{{u}}=\hat{{v}}\Leftrightarrow0={d}\left({u},\,{v}\right)=\sum_{i}{d}_{i}\left({u}_{i},\,{v}_{i}\right)\Leftrightarrow\forall{i}\in\overline{1,\,{n}}\ {d}\left({u}_{i},\,{v}_{i}\right)\Leftrightarrow\forall{i}\in\overline{1,\,{n}}\ \hat{{u}_{i}}=\hat{{v}_{i}}\)
\qed \end{proof}

\begin{lemma}Let \({s}\in\Sigma^*\) and \({s}^\prime\) be the string obtained by replacing all the occurences of a character \(\alpha\in\Sigma\) in \({s}\) with \(\beta\in\Sigma\). Then \(\Delta_{s}\subseteq\Delta_{{s}^\prime}\).
\end{lemma}

\begin{proof}
\(\Delta_{s}=\lbrace\delta\in\mathbb{N}^*\vert\forall{j}\in\overline{1,\,\delta}\ {s}_{j}={s}_{\lvert s\rvert -\delta +j}\rbrace=
\lbrace\delta\in\mathbb{N}^*\vert\forall{j}\in\overline{1,\,\delta}\ \left({s}_{j}={s}_{\lvert s\rvert -\delta +j}=\alpha\lor{s}_{j}={s}_{\lvert s\rvert -\delta +j}=\beta\right)\lor{s}_{j}={s}_{\lvert s\rvert -\delta +j}\neq\alpha,\,\beta\rbrace\subseteq
\lbrace\delta\in\mathbb{N}^*\vert\forall{j}\in\overline{1,\,\delta}\ {s}^\prime_{j}={s}^\prime_{\lvert s^\prime\rvert -\delta +j}=\beta\lor{s}^\prime_{j}={s}^\prime_{\lvert s^\prime\rvert -\delta +j}\neq\alpha,\,\beta\rbrace=
\lbrace\delta\in\mathbb{N}^*\vert\forall{j}\in\overline{1,\,\delta}\ {s}^\prime_{j}={s}^\prime_{\lvert s^\prime\rvert -\delta +j}\rbrace=\Delta_{{s}^\prime}
\)
\qed \end{proof}

\begin{corollary}Let \(\mathcal{I}\in\mathcal{T}_{n}\), \({s},\,{s}^\prime\) as above. Then \({s}\in\Sigma_\mathcal{I}\Rightarrow{s}^\prime\in\Sigma_\mathcal{I}\).
\end{corollary}

\begin{theorem}Let \(\left(\Sigma,\,{d}\right)\) be a pseudometric space. Then SCP/ACP for the additive pseudometric over \(\Sigma^{n}\) is equivalent with the MSC/ACP for the addive metric over \(\left(\Sigma/\approx \right)^{n}\), modulo \(\mathcal{O}\left(\lvert{w}\rvert\lvert\Sigma\rvert\right)\) work.
\end{theorem}
\begin{proof}
Let \(\left({s}^*,\,\mathcal{I}^*\right)\) be the ACP solution for \({w}\) with respect to the pseudometric \({d}\), \(\Sigma/\approx{\ }\ni\hat{{x}}\rightarrow\phi\left(\hat{{x}}\right)\in\hat{{x}}\subseteq\Sigma\), \(\psi=\phi\circ\hat{\ }\) and \({s}^{*^\prime}=\psi\left({s}^*\right)\). We have that \({s}^{*^\prime}\in\Sigma_{\mathcal{I}^*}\) and moreover \({d}\left(\mathcal{I}^*\left({s}^{*^\prime}\right),\,{w}\right)=-{d}\left(\mathcal{I}^*\left({s}^*\right),\,\mathcal{I}^*\left({s}^{*^\prime}\right)\right)+{d}\left(\mathcal{I}^*\left({s}^{*^\prime}\right),\,{w}\right)\leq{d}\left(\mathcal{I}^*\left({s}^*\right),\,{w}\right)\leq{d}\left(\mathcal{I}^*\left({s}^*\right),\,\mathcal{I}^*\left({s}^{*^\prime}\right)\right)+{d}\left(\mathcal{I}^*\left({s}^{*^\prime}\right),\,{w}\right)={d}\left(\mathcal{I}^*\left({s}^{*^\prime}\right),\,{w}\right)\). This is true for any other candidate string, not just for the optimum, and hence any solution for the MSC/ACP \(\left(\Sigma/\approx{\ }\right)^{n}\) leads to a solution MSC/ACP for the addive metric over \(\Sigma^{n}\) and vice versa, one via \(\hat{\ }\circ\psi\) and the other via \({x}\rightarrow\hat{{x}}\). Tabulating the functions themselves can be done  in \(\mathcal{O}\left(\lvert{w}\rvert\lvert\Sigma\rvert\right)\).
\qed \end{proof}

\section{Block Variations}
\label{sec:block}

Consider once again the situation where transactional data is to be analyzed by metadata. A transaction is a string of operations just like a message is a string of characters. Theoretically, if the idle operation is a valid one, then the number of transactions is unbounded even though it may be represented over a bounded alphabet of operations. We would like to investigate how the complexity of the \textbf{ACP} changes when we take a lower-level approach and switch the data representation. Naturally, if we want to represent the same patterns we have to consider an equivalent metric.

\begin{theorem}
\label{thm:induced-metric_art}
Let $\left(X,\,d\right)$ be an (extended) (pseudo-)metric space and $\phi:X\rightarrow Y$ an injection. Then there exists \(d^\prime:Y\times Y\rightarrow\bar{\mathbb{R}}_+\) such that $\left(Y,\,d^\prime\right)$ is an extended (pseudo-)metric space and moreover $d^\prime_{\vert\phi\left(X\right)\times\phi\left(X\right)}=d$ and $d^\prime_{\vert\phi\left(X\right)\times\left(Y\backslash\phi\left(X\right)\right)}=\infty$. We say that $d^\prime$ is the metric induced by $\phi$ on $Y$.
\end{theorem}

\begin{proof}
We define \({d}^\prime:{Y}\times{Y}\rightarrow\bar{\mathbb{R}}_+\) where \({Y}\times{Y}\ni\left({x},\,{y}\right)\rightarrow{d}^\prime\left({x},\,{y}\right)=\left(1-\delta_{{x},\,{y}}\right){d}^\prime\left({y},\,{x}\right)\), where \(\delta\) is the Kronecker delta, such that \({X}\times{X}\ni\left({x},\,{x}\right)\rightarrow{d}\left({x},\,{y}\right)={d}^\prime\left(\phi\left({x}\right),\,\phi\left({y}\right)\right)\) and \(\left({Y}\backslash\phi\left({X}\right)\right)\times\left({Y}\backslash\phi\left({X}\right)\right)\ni\left({x},\,{y}\right)\rightarrow{d}^\prime\left({x},\,{y}\right)=
\begin{cases}
0 & {x}={y} \\
\infty & {x}\neq{y}
\end{cases}
\) and so we  have that \(\left({Y},\,{d}^\prime\right)\) is an extended (pseudo-)metric space and that the indiscernible pairs of distinct points in \({Y}\) with respect to \({d}^\prime\) are exactly the images of the indiscernible pairs of distinct points in \({X}\) with respect to \({d}\). 
\qed \end{proof}

Recall that a (pseudo)metric is a map from the pairs of points of a space $X$ to the reals satisfying some axioms. An extended (pseudo)metric is a map from the pairs of points of a space $X$ to the extended reals, $\overline{\mathbb{R}}$ (thus allowing $\infty$) satisfying the same respective axioms.

This may appear counterintuitive, since we would like to process our unbounded alphabet, but if $X$ were unbounded so would be $Y$. Let $\Sigma$ be a finite alphabet. If we have an injection $\psi:\Sigma \rightarrow X$, but not the other way around, there always exists some power $n$ of $\Sigma$ such that there is an injection $\phi:X\rightarrow\Sigma^n$, where $n$ is unbounded. Hence, there exists an injection $\phi: X\rightarrow \Sigma^*$. 

Note that a bijection is not required, which is beneficial since we may not always find one. For example, even though all messages sent over a network can be represented as strings of bytes, their representation is in general not bijective due to some redundancies such as the CRC. In case a bijection $\phi$ does however exist, it is called a translation and $d^\prime$ is the translated metric.

\begin{definition}[Translation]
If \(\phi\) is bijective we say that the (approximate) string cover problem over \(\Gamma^*\) with respect to \(\delta^\prime\) is a translation of the (approximate) string cover problem over \(\Sigma^*\) with respect to \(\delta\). 
\end{definition}

An injection $\phi:X\rightarrow Y$ naturally lifts to $\phi:X^*\rightarrow Y^*$ over strings, but not necesarily to an injection. This is particularily important when $Y=\Sigma^*$ are strings themselves. In many comunication protocols a problem is splitting the flux into telegrams. If all telegrams have the same size i.e. $Y=\Sigma^n$ such a split is easy to do and the lifted $\phi:X^*\rightarrow\left(\Sigma^n\right)^*$ is naturally injective. However, when the telegrams have variable length we have to induce some additional structure, such as a $wake-up$ bit (like in MARK/SPACE serial protocols), a $terminal$ character (like with C strings) or start a string with its length (like with Pascal strings and \texttt{std::string}).



\begin{definition}[Block Variations]
If $Y=\Sigma^n$ for some $n$ or $Y=\Sigma^*\bullet$ such that $\lvert Y \rvert>\lvert X\rvert>\lvert\Sigma\rvert,\,\bullet\notin\Sigma$ we say that the (approximate) string cover problem over $Y^*$ with respect to $d^\prime$ is a block variation of the (approximate) string cover problem over $X^*$ with respect to $d$. 
\end{definition}

\begin{theorem}
\label{thm:block-complexity}
The block variation of an ACP/SCP has the same complexity as the original modulo $\mathcal{O}\left(Nf(1)+g(N)\right)$ where $N$ is the length of the input, $f$ is the complexity of $\phi:X^*\rightarrow Y^*$ and $g$ the complexity of \(\phi^{-1}:Y^*\rightarrow X^*\)
\end{theorem}

\begin{remark}
Using G\"odel's encoding ($\phi\left(w\right)=\underset{i=1}{\overset{\lvert w\rvert}\prod}p_i^{w_i}$ in base 1 where $p_i$ is the $i^{th}$ prime number) we can represent a string using a single character. We cannot however lift the telegrams produced this way using a Pascal encoding, and we require one of the other two. Thus, block variations only make sense over alphabets that are at least binary.
\end{remark}

\begin{theorem}
\label{thm:new-metric}
For each fixed ${p}\geq2$, there is a metric $d^\prime$ on ${\left(\mathbb{Z}_p\right)}^*$ induced by a metric $d$ on the naturals, and thus by any restriction of that metric to a $\mathbb{Z}_n$, obtained by an injection for which both it and its inverse are computable in linear time with regard to the input size and logarithmic time with regard to $p$.
\end{theorem}

\begin{corollary}
For any metric, the complexity of the ACP/SCP on an unbounded size alphabet is the same as the complexity of its block variation induced by such a function as the one in the theorem above on a finite size alphabet, using extra $\mathcal{O}\left(N\log p\right)$ time.
\end{corollary}

\bibliographystyle{bababbrv}
\bibliography{bibliography}

\newpage

\appendix

\section{Omitted proofs}

\begin{proof}[of Theorem~\ref{thm:block-complexity}]
Let O be the complexity of the original problem and BV the complexity of its block variation.

Let \({w}\in\Sigma^*\). After \(\mathcal{O}\left({f(N}\right)=\mathcal{O}\left({Nf(1)}\right)\) work we can solve the block variation and hence \({O}\leq{BV}+\mathcal{O}\left({Nf(1)}\right)\).

Let \({w}\in\Gamma^*\). After \(\mathcal{O}\left({g(N}\right)\) work we can solve the original and hence \({BV}\leq{O}+\mathcal{O}\left({g(N)}\right)\). Note that if there is no inverse there is no solution to the ACP and the only solution to the SCP is the string itself.
\qed \end{proof}

\begin{proof}[of Theorem~\ref{thm:new-metric}]
We are now on the lookout for a \(\psi\) that is more "variable-size" in nature. The most natural example is defined recursively. Let \(\tau_{p-2}:\mathbb{N}\rightarrow{\left(\mathbb{Z}_{p-2}\right)}^*\) be the base-(p-2) conversion i.e. \(\lvert\tau_{p-2}\left({x}\right)\rvert=\lceil\log_{p-2}\left({x+1}\right)\rceil\) with \(\tau_{p-2}\left({x}\right)_{\lceil\log_{p-2}\left({x+1}\right)\rceil+1-{i}}=\left({x}-\underset{{j}=1}{\overset{{i}-1}\sum}{\tau_{p-2}\left({x}\right)}_{\lceil\log_{p-2}\left({x+1}\right)\rceil+1-{j}}\ \textmd{mod}\ \left({p-2}\right)\right)\ \forall{i}\in\overline{1,\,\lceil\log_{p-2}\left({x+1}\right)\rceil}\) and consider the much less trivial: \newline\(\psi({x})=\begin{cases}
0\left({p}-1\right)\tau_{p-2}\left({n}\right) & \lceil\log_{p-2}\left({x+1}\right)\rceil\leq{p-3} \\
\psi\left(\lceil\log_{p-2}\left({x+1}\right)\rceil\right)\left({p}-1\right)\tau_{p-2}\left({n}\right) & \lceil\log_{p-2}\left({x+1}\right)\rceil>{p-3}
\end{cases}\)

Note that this satisfies \(\lvert\psi\left({x}\right)\rvert\leq\lvert\psi\left(\lceil\log_{p-2}\left({x+1}\right)\rceil\right)\rvert+\lceil\log_{p-2}\left({x+1}\right)\rceil\leq{x}+1\Rightarrow\lvert\psi\left({x}\right)\rvert\leq2\lceil\log_{p-2}\left({x+1}\right)\rceil\Rightarrow\lvert\psi\left({x}\right)\rvert\in\mathcal{O}\left(\log_{p}{n}\right)\) and thus this is also a logarithmic-time logarithmic deflation. What remains now is to give a way to compute \({d}^{\prime\prime}\) in \(\mathcal{O}\left({N}\log_{p}{n}\right)\), or, equivalently, a way to compute \(\phi^{-1}:{\left(\mathbb{Z}_{p-1}\right)}^*\rightarrow\mathbb{N}\) in \(\mathcal{O}\left(\log_{p}{n}\right)\). Let \(x_1<x_2<\dots<x_{k}\) be the set of apparitions of \(\left({p}-1\right)\) in \({w}\in{\left(\mathbb{Z}_{p-1}\right)}^*\). \(\psi^{-1}\left({w}\right)\) exists if and only if \(x_1>1= x_0+1,\,x_k<\lvert{w}\rvert= x_{k+1}-1\) and \(\forall{i}\in\overline{1,\,{k}}\ \tau_{p-2}^{-1}\left({w}_{x_{i-1}+1}\dots{w}_{x_{i}-1}\right)=\lceil\log_{p}\left(1+\tau_{p-2}^{-1}\left({w}_{x_{i}+1}\dots{w}_{x_{i+1}-1}\right)\right)\rceil\), in which case \(\psi^{-1}\left({w}\right)=\tau_{p-2}^{-1}\left({w}_{x_{k}+1}\dots{w}_{x_{k+1}-1}\right)\) so this can  be checked in logarithmic time.
\qed \end{proof}

\end{document}